\documentclass{article} 
\usepackage{authblk}
\usepackage{graphicx} 
\usepackage{amssymb}
\usepackage{amsthm}
\usepackage{amsmath}
\usepackage[numbers]{natbib}
\usepackage[dvipsnames]{xcolor}
\usepackage{enumerate}
\usepackage[shortlabels]{enumitem}

\usepackage[normalem]{ulem}
\usepackage{comment}

\newcommand{\ch}{{\rm ch}}
\newcommand{\ind}{{\rm ind}}
\newcommand{\rSPR}{\mathrm{rSPR}}
\newcommand{\switch}{\mathrm{switch}}

\newcommand{\pa}{\mathrm{pa}}
\newcommand{\sw}{\mathrm{sw}}

\newcommand{\cl}{\mathrm{cl}}

\newtheorem{theorem}{Theorem}[section]
\newtheorem{lemma}[theorem]{Lemma}
\newtheorem{corollary}[theorem]{Corollary}

\newtheorem{observation}[theorem]{Observation}

\title{Bounding the softwired parsimony score of a phylogenetic network}
\author[1]{Janosch D\"ocker}
\author[1]{Simone Linz}
\author[2]{Kristina Wicke}
\affil[1]{School of Computer Science, University of Auckland, New Zealand}
\affil[2]{Department of Mathematical Sciences, New Jersey Institute of Technology, USA}
\date{}

\begin{document}

\maketitle

\begin{abstract}
In comparison to phylogenetic trees, phylogenetic networks are more suitable to represent complex evolutionary histories of species whose past includes reticulation such as hybridisation or lateral gene transfer. However, the reconstruction of phylogenetic networks remains challenging and computationally expensive due to their intricate structural properties. For example, the small parsimony problem that is solvable in polynomial time for phylogenetic trees, becomes NP-hard on phylogenetic networks under softwired and parental parsimony, even for a single binary character and structurally constrained networks. To calculate the parsimony score of a phylogenetic network $N$, these two parsimony notions consider different exponential-size sets of phylogenetic trees that can be extracted from $N$ and infer the minimum parsimony score over all trees in the set. 
In this paper, we ask: What is the maximum difference between the parsimony score of any phylogenetic tree that is contained in the set of considered trees and a phylogenetic tree whose parsimony score equates to the parsimony score of $N$?
Given a gap-free sequence alignment of multi-state characters and a rooted binary level-$k$ phylogenetic network, we use the novel concept of an informative blob to show that this difference is bounded by $k+1$ times the softwired parsimony score of $N$. In particular, the difference is independent of the alignment length and the number of character states. We show that an analogous bound can be obtained for the softwired parsimony score of semi-directed networks, while under parental parsimony on the other hand, such a bound does not hold.
\end{abstract}


\section{Introduction}
The generalisation of phylogenetic trees to phylogenetic networks goes along with the development of new methods to reconstruct phylogenetic networks from genomic data. Since phylogenetic networks are structurally much more complicated than phylogenetic trees, the  algorithms to infer networks are typically computationally  expensive. Indeed, several optimisation problems that can be solved efficiently for phylogenetic trees, become computationally difficult for phylogenetic networks. For example, the small parsimony problem, which seeks to find the {\it parsimony score} of a given phylogenetic tree with character states assigned to its leaves is solvable in polynomial time for phylogenetic trees, e.g., using the well-known Fitch-Hartigan algorithm~\cite{fitch1971toward,hartigan1973algorithm}, but becomes NP-hard under different notions of parsimony for phylogenetic networks, even for a single binary character and structurally constrained networks~\cite{fischer2015computing,van2018improved}.

Despite the popularity of model-based methods to infer phylogenetic trees, maximum parsimony (see, e.g.~\cite{felsenstein2004inferring} and references therein) continues to be widely used in certain areas of evolutionary biology, such as the analysis of morphological data (see, e.g. \cite{sansom2018parsimony,schrago2018comparative,smith2019bayesian}). Moreover, since calculating the parsimony score of a phylogenetic tree is computationally less expensive than calculating its likelihood, parsimony trees are often used as starting trees from which a search through tree space is started~\cite{stamatakis2005raxml} and are also used in Bayesian phylogenetic inference~\cite{zhang2020parsimonyguided}.

Recently,  different notions for parsimony on rooted phylogenetic networks have been proposed, referred to as  {\it hardwired}, {\it softwired}, and {\it parental} parsimony. Softwired~\cite{nakhleh2005reconstructing} and parental~\cite{van2018improved} parsimony both consider collections of trees that can be extracted from a rooted phylogenetic network (so-called {\it displayed trees}, respectively {\it parental trees}) and define the parsimony score as the minimum parsimony score of any tree in the collection. 
Softwired parsimony is implemented in the popular software package \textsf{PhyloNetworks}~\cite{solis2017phylonetworks} and is the main focus of this paper.
Hardwired parsimony~\cite{kannan2012maximum}, on the other hand, calculates the parsimony score of a phylogenetic network by considering character-state transitions along all edges of the network. As the sets of rooted phylogenetic trees that are evaluated when computing the softwired and parental parsimony scores of a phylogenetic network have exponential size, it is of interest to investigate the differences in parsimony scores of elements in these sets. 
Given a gap-free alignment of multi-state characters and a rooted binary level-$k$ network $N$ (formally defined below), we analyse how different the parsimony score of any phylogenetic tree displayed by $N$ and the softwired parsimony score of $N$ can be. We show that independent of the alignment length and number of character states, this difference is bounded by $k+1$ times the parsimony score of $N$. Thus, while computing the softwired parsimony score is in general an NP-hard problem (even for a single binary character and a structurally constrained network~\cite[Theorem 4.3]{fischer2015computing}), our result implies that an upper bound for the softwired parsimony score of $N$ can be obtained in polynomial time by simply evaluating the parsimony score of an arbitrary phylogenetic tree that is displayed by $N$. In particular if the level of $N$ is small, this upper bound gives a good indication of the magnitude of the softwired parsimony score of $N$. 

Related to our result, it was shown by Fischer et al.~\cite[Theorem 5.7]{fischer2015computing} that the NP-hard optimisation problem of computing the softwired parsimony of a rooted level-$k$ network for a single multi-state character is, on the positive side, also fixed-parameter tractable, when the parameter is $k$. If one considers more than a single binary character, the softwired parsimony problem is NP-hard even for a rooted level-$1$ network~\cite[Theorem 1]{kelk2019finding}. As a consequence of the latter negative result, Kelk et al.~\cite{kelk2019finding} posed the following question.  Are there good (i.e.~constant-factor) approximation algorithms for computing the softwired parsimony score of a rooted phylogenetic network $N$ and a sequence alignment $A$ without gaps under the following three restrictions: (i) $N$ is level-$1$, (ii) each biconnected component of $N$ has exactly three outgoing edges, and (iii) $A$ consists of binary characters?

As hinted at above, from an algorithmic perspective, our upper bound result implies a $(k+1)$-approximation algorithm for computing the softwired parsimony score of a rooted binary level-$k$ network $N$. Specifically, take an arbitrary phylogenetic tree that is displayed by $N$, compute its parsimony score, and use this to approximate the softwired parsimony score of $N$. If the level of $N$ is fixed, this algorithm provides a polynomial-time constant factor approximation. Hence, we answer the aforementioned question by Kelk et al. affirmatively for a much larger class of rooted phylogenetic networks in the sense that our result holds for level-$k$ networks (for a fixed non-negative integer $k$), it does not require restriction (ii), and it holds for gap-free alignments independent of the number of character states. Our result also complements a recent paper by Frohn and Kelk~\cite{frohn23}, in which the authors establish a 2-approximation algorithm for the softwired parsimony problem on binary tree-child  networks for a single character. 

While softwired parsimony for rooted phylogenetic networks is the main focus of our paper, we additionally show that an analogous upper bound for the softwired parsimony score holds for {\it semi-directed} networks that are obtained from rooted phylogenetic networks by deleting the root and omitting the direction of all but reticulation edges. Semi-directed networks have recently been central to studying identifiability questions related to phylogenetic networks and to developing phylogenetic network estimation algorithms (e.g.~\cite{allman2019nanuq,gross2018distinguishing,hollering2021identifiability,solis2016inferring}). We also briefly turn to the notion of parental parsimony (on rooted phylogenetic networks) and show by way of counterexample that an analogous bound for the parental parsimony score does not hold. 

The remainder of this paper is organised as follows. We define all relevant concepts related to phylogenetic trees and networks, introduce the notion of softwired parsimony, and state the main result in Section~\ref{sec:prelim}. In Section~\ref{sec:rSPR}, we revisit the rSPR distance and establish an upper bound on this distance for two phylogenetic trees that are both displayed by a given phylogenetic network. In Sections~\ref{sec:informativeblobs} and~\ref{sec:clusterreduction}, we introduce the notion of an {\it informative blob} and a {\it blob reduction}, respectively. Informative blobs are a novel concept that is crucial for obtaining our main result, the upper bound on the softwired parsimony score, which we establish in Section~\ref{sec:main}. Sections~\ref{sec:parental}  and \ref{sec:semi-directed} are then devoted to parental parsimony on rooted networks and softwired parsimony on semi-directed networks, respectively. We end the paper with some concluding remarks and directions for future research in Section~\ref{sec:conclusion}.

\section{Preliminaries and statement of main result}\label{sec:prelim}
This section introduces notation and terminology, and states the main result. Throughout this paper, $X$ denotes a non-empty finite set. Let $G$ be a directed graph. We use $V(G)$ and $E(G)$ to denote the vertex set and edge set, respectively, of $G$. Furthermore, for each edge $(u,v)$ of $G$, $u$ is called a {\it parent} of $v$ and $v$ is called a {\it child} of $u$. We sometimes also refer to $u$ and $v$ as {\it neighbours} in $G$. In a similar vein, for two (not necessarily distinct) vertices $s$ and $t$ of $G$, we say that $s$ (resp. $t$) is an {\it ancestor} (resp. {\it descendant}) of $t$ (resp. $s$) if there is a directed path of length zero or more from $s$ to $t$. Now let $G$ and $G'$ be two directed graphs, and let $e=(u,w)$ be an edge of $G$. Then {\it subdividing} $e$ is the operation of replacing $e$ with two new edges $(u,v)$ and $(v,w)$. Furthermore, we call $G'$ a {\it subdivision} of $G$ if $G'$ can be obtained from $G$ by repeatedly subdividing an edge. We also consider $G$ to be a subdivision of itself.  \\

\noindent{\bf Phylogenetic trees and networks.} 
A {\it rooted binary phylogenetic network $N$ on $X$}  is a rooted acyclic digraph with no loops and no parallel edges that satisfies the following three properties:
\begin{enumerate}[label={\upshape(\roman*)}]
\item the set of leaves is $X$,
\item the out-degree of the (unique) root $\rho$
is exactly one, and 
\item every other vertex has either in-degree one and out-degree two, or in-degree two and out-degree one.
\end{enumerate}
The set $X$ is also sometimes called the {\it label set} of $N$. Furthermore, a vertex of $N$ is referred to as a {\it reticulation} if it has in-degree two and as a {\it tree vertex} if it has in-degree one and out-degree two. Similarly, an edge of $N$ that is directed into a reticulation is referred to as a {\it reticulation edge}. We denote the number of reticulations in $N$ by $h(N)$. 

Let $N$ be a rooted binary phylogenetic network on $X$. If $N$ has no reticulation, then it is called a {\it rooted binary phylogenetic $X$-tree}. Since all  phylogenetic trees and networks are rooted and binary throughout this paper except for Section~\ref{sec:semi-directed}, we refer to a rooted binary phylogenetic network as a {\it phylogenetic network on $X$} and to a rooted binary phylogenetic tree as a {\it phylogenetic $X$-tree}. 

Let $S$ be a subdivision of a phylogenetic $X$-tree. We call the directed path from the root of $S$ to its closest degree-three vertex its {\it root path}. If $S$ is a phylogenetic tree, then the root path consists of a single edge, in which case we sometimes refer to the root path  as the {\it root edge.}

Now let $N$ be a phylogenetic network. A {\it biconnected component} of $N$ is a maximal subgraph of $N$ that is connected and cannot be disconnected by deleting exactly one of its vertices. Furthermore, a vertex of a biconnected component of $N$ is called a {\it reticulation} if it is a reticulation in $N$. With this definition in hand, we say that $N$ is {\it level-$k$} if the maximum number of reticulations of a biconnected component of $N$ is at most $k$. Lastly, we call a biconnected component of $N$ a {\it blob} if it has at least one reticulation. For a blob $B$ of $N$, we refer to the unique vertex with in-degree zero and out-degree two in $B$ as the {\it source} of $B$. A phylogenetic network $N$ on $\{x_1,x_2,\ldots,x_8\}$ with two blobs $B$ and $B'$ is shown on the left-hand side of Figure~\ref{fig:informative-blob}.\\

\noindent{\bf Clusters.}
Let $M$ be a subdivision of a phylogenetic network on $X$, and let $Y$ be a subset of $X$. We call $Y$ a {\it cluster} of $M$ if there exists a vertex $v$ in $M$ that has precisely $Y$ as its set of descendant leaves. Note that there may be more than one vertex in $M$ whose cluster is $Y$ and that this may also be the case if $M$ is a subdivision of a phylogenetic $X$-tree.  Furthermore, we use $\cl_M(v)$ or $\cl(v)$ if the subscript is clear from the context to denote the cluster of a given vertex $v$ of $M$.\\

\noindent{\bf Displaying.}
Let $N$ be a phylogenetic network on $X$ with root $\rho$, and let $T$ be a phylogenetic $X'$-tree with $X'\subseteq X$. We say that $T$ is {\it displayed} by $N$ if there exists a subgraph of $N$ that is a subdivision of $T$ that includes $\rho$, in which case this subgraph is called an {\it embedding} of $T$ in $N$. 
The set of all  phylogenetic $X$-trees that are displayed by $N$ is referred to as the {\it display set} of $N$ and denoted by $D(N)$. Ignoring the assignment of 0 and 1 to vertices for the moment, Figure~\ref{fig:informative-blob} shows a phylogenetic network $N$, a phylogenetic tree $T$ that is displayed by $N$, and an embedding $S$ of $T$ in $N$. Now, consider a subset $R$ of the reticulation edges of $N$. We refer to $R$ as a {\it switching} if, for each reticulation $v$ in $N$, it contains exactly one of the two edges that are directed into $v$. By deleting each reticulation edge of $N$ that is not in $R$, we obtain a connected subgraph $G$ of $N$ with no underlying cycle and, for each leaf $\ell\in X$, there is a directed path from the root of $G$, which coincides with $\rho$, to $\ell$.
If we repeatedly suppress each vertex in $G$ with in-degree one and out-degree one, and delete each vertex in $G$ with out-degree zero that is not in $X$ until no such operation is possible, we obtain a phylogenetic $X$-tree $T_R$. We say that $R$ yields $T_R$. By construction, $T_R$ is displayed by $N$. Conversely, observe that, for each phylogenetic $X$-tree $T$ in $D(N)$, there exists at least one switching that yields $T$. In summary, we have the following observation.

\begin{figure}[t]
    \centering
    \includegraphics[width=\textwidth]{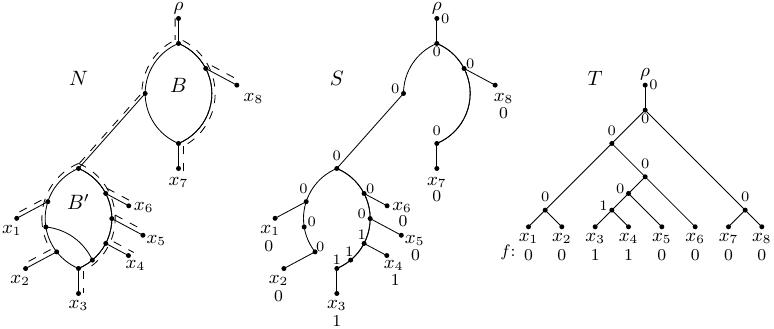}
    \caption{A phylogenetic network $N$ (left), an embedding $S$ of a phylogenetic $X$-tree displayed by $N$ together with a binary character $f$ and an extension $F$ (middle; also indicated by the dashed lines in the left panel), and the phylogenetic $X$-tree $T$ obtained from $S$ by suppressing vertices of in-degree one and out-degree one (right).}
    \label{fig:informative-blob}
\end{figure}

\begin{observation}\label{ob:switching}
Let $N$ be a phylogenetic network on $X$, and let $T$ be a phylogenetic $X$-tree. Then $T$ is displayed by $N$ if and only if there exists a switching of $N$ that yields $T$.
\end{observation}

\noindent {\bf rSPR operation.} Let $T$ be a phylogenetic $X$-tree, and let $e=(u,v)$ be an edge of $T$ that is not incident with the root. Let $T'$ be a phylogenetic $X$-tree obtained from $T$ by deleting $e$ and reattaching the resulting rooted subtree that contains $v$ via a new edge $f$ in the following way: Subdivide an edge of the component that contains the root of $T$ with a new vertex $u'$, join $u'$ and $v$ with $f$, and suppress $u$. We say that $T'$ has been obtained from $T$ by a {\it rooted subtree prune and regraft} (rSPR) operation. The {\it rSPR distance} between any two phylogenetic $X$-trees $T$ and $T'$, denoted by $d_\rSPR(T,T')$, is the minimum number of rSPR operations that transform $T$ into $T'$. It is well known that $T'$ can always be obtained from $T$ by a sequence of single rSPR operations and, so, $d_\rSPR(T,T')$ is well defined.\\

\noindent{\bf Characters.} 
An {\it $r$-state character on $X$} is a surjective function $f\colon X\rightarrow C$ from $X$ into a set $C$ of character states with $r=|C|\geq 1$. If $r=2$, then $f$ is called a {\it binary} character. Throughout this paper, all results are established for $r$-state characters with $r$ being fixed and arbitrarily large. For simplicity, we refer to an $r$-state character on $X$ as a {\it character} on $X$. 

Let $G$ be an acyclic digraph with leaf set $X$, and let $f\colon X\rightarrow C$ be a character on $X$. An {\it extension} of $f$ to $V(G)$ is a function $F\colon V(G)\rightarrow C$ such that $F(\ell)=f(\ell)$ for each element $\ell\in X$. For an extension $F$ of $f$, we set 
$$\ch(F,G)=|\{(u,v)\in E(G):F(u)\ne F(v)\}|,$$ and refer to $\ch(F,G)$ as the {\it changing number} of $F$.
Intuitively, each edge of $G$ that contributes to the changing number of $F$ requires a character-state transition to explain $f$ on $G$. Lastly, we say that an extension $F$ of $f$ to $V(G)$ is {\it minimum} if there exists no extension of $f$ to $V(G)$ whose changing number is strictly smaller than that of $F$.

In what follows, we often consider a sequence $(f_1,f_2,\ldots,f_n)$ of characters on $X$ instead of a single character. We call such a sequence an 
{\it alignment}. Unless stated otherwise, all alignments in this paper are sequences of $r$-state characters for $r\geq 2$ that do not contain the gap symbol ``--''. Such an alignment is referred to as {\it gap-free}. In applied phylogenetics, multiple sequence alignments frequently contain gaps which, intuitively, are placeholders that can take on any of the other $r$ character states. We will see in the last section why the restriction to gap-free alignments is necessary. Lastly, we denote a sequence $(f_1)$ that consists of a single element by $f_1$ and omit parentheses for simplicity.\\

\noindent{\bf Parsimony on phylogenetic trees and their subdivisions.}
Given an alignment $A=(f_1,f_2,\ldots,f_n)$ of characters on $X$ and an arbitrary rooted tree $T$ with leaf set $X$, we refer to 
$$PS(A,T)=\sum_{i=1}^n \min_{F_i}(\ch(F_i,T))$$
as the {\it parsimony score} of $A$ on $T$, where the minimum is taken over all extensions of $f_i$ to $V(T)$. 

Instead of calculating the parsimony score of a phylogenetic $X$-tree $T$, we are often interested in calculating the parsimony score of a subdivision of $T$ in the upcoming sections. The next lemma states that both scores are equal. Its correctness can be established analogously to the proof of Lemma 4.5 in~\cite{fischer2015computing}. In the proof of this lemma, Fischer et al. have shown that the softwired parsimony score of a character $f$ on a phylogenetic tree $T$ is equal to the parsimony score of $f$ on a particular rooted tree that is a generalisation of a subdivision of $T$ in the sense that it may contain  unlabeled leaves in addition to the leaves in $X$.

\begin{lemma}\label{l:fitch}
Let $f$ be a character on $X$, and let $S$ be a subdivision of a phylogenetic $X$-tree $T$. Then $PS(f,S)=PS(f,T)$.
\end{lemma}

We also have the following observation.

\begin{observation}\label{ob:root-edge}
Let $f$ be a character on $X$, and let $S$ be a subdivision of a phylogenetic $X$-tree. If $F$ is an extension of $f$ to $V(S)$ such that $F(u)\ne F(v)$ for some edge $(u,v)$ of the root path of $S$, then there exists an extension $F'$ of $f$ to $V(S)$ such that $F'(u)=F'(v)$ and $\ch(F',T)<\ch(F,T)$.
\end{observation}

\noindent By Observation~\ref{ob:root-edge}, we  freely assume throughout the remainder of the paper that  every extension  $F$ of a character to the vertices of a subdivision of a phylogenetic tree has the additional property that there is no character state transition on any edge of its root path. \\

\noindent{\bf Parsimony on phylogenetic networks.}
As outlined in the introduction, several notions of parsimony have been introduced that generalise 
parsimony from phylogenetic trees to phylogenetic networks. In this paper, we are focusing on the notion of softwired parsimony and briefly touch on parental parsimony in Section~\ref{sec:parental}. Roughly speaking, the softwired parsimony score of an alignment $A$ of characters on a phylogenetic network $N$ is defined to be the smallest number of character-state transitions that is necessary to explain $A$ on any phylogenetic tree that is displayed by $N$. Following~\citet{nakhleh2005reconstructing}, we now make this precise. Let $A=(f_1,f_2,\ldots,f_n)$ be an alignment of characters on $X$, and let $N$ be a phylogenetic network on $X$. The {\it softwired parsimony score} of $A$ on $N$ is defined as 
\begin{align}\label{eq:soft-one}
PS_{\sw}(A,N)=\sum\limits_{i=1}^n PS_{\sw}(f_i, N)&=\sum_{i=1}^n\min_{T\in D(N)}\min_{F_i}(\ch(F_i,T))\nonumber\\
&=\sum_{i=1}^n\min_{T\in D(N)}PS(f_i,T),
\end{align}
where, for each character $f_i$, the first minimum is taken over all  phylogenetic trees in the display set of $N$ and the second minimum is taken over all extensions of $f_i$ to $V(T)$. As per Equation (\ref{eq:soft-one}), each character in $A$ can follow a different tree in $D(N)$. A slightly more restricted definition of softwired parsimony, which has also appeared in the literature (e.g. see~\cite{kelk2017complexity,kelk2019finding}), is the following
\begin{align}\label{eq:soft-two}
     PS_{\sw'}(A,N)= \min\limits_{T \in D(N)} \sum\limits_{i=1}^n \min_{F_i}(\ch(F_i,T))
     = \min\limits_{T \in D(N)} \sum\limits_{i=1}^n PS(f_i,T),
\end{align}
where all characters in $A$ follow the same tree in $D(N)$. Clearly, if $n=1$, then  $PS_\sw(A,N)=PS_{\sw'}(A,N)$. On the other hand, for $n\geq 1$, it follows from the definition that $PS_\sw(A,N)\leq PS_{\sw'}(A,N)$. For the purpose of the upcoming sections, we adopt the softwired parsimony definition as formalised in Equation (\ref{eq:soft-one}) and will see later, that our main result holds also under the definition given in Equation (\ref{eq:soft-two}). In this context, it is worth mentioning that the hardness result for computing the softwired parsimony score of a level-$1$ network for an alignment of at least two binary characters~\cite[Theorem 1]{kelk2019finding} as mentioned in the introduction has been established for the definition given in Equation~\eqref{eq:soft-two}.\\

\noindent{\bf Statement of main result.}
The main result of this paper is the following theorem which we establish in Section~\ref{sec:main}.

\begin{theorem}\label{t:main}
Let $N$ be a phylogenetic network on $X$, and let $T$ be a phylogenetic $X$-tree in $D(N)$. Furthermore, let $A$ be an alignment of characters on $X$. Then 
$$PS(A,T)\leq (k+1)\cdot PS_{\sw}(A,N) \text{ and }PS(A,T)\leq (k+1)\cdot PS_{\sw'}(A,N),$$ 
where $k$ is the level of $N$. 
\end{theorem}

\noindent For example, if $N$ is a level-1 network, Theorem~\ref{t:main} implies that the parsimony score of an arbitrary tree displayed by $N$ is at most twice the parsimony score of a tree displayed by $N$ whose parsimony score is equal to $PS_{\sw}(A,N)$. Moreover, we show in Section~\ref{sec:main} that the bound as stated in Theorem~\ref{t:main} is sharp.

The next corollary positively answers the open problem that is detailed in the introduction and that was first posed in~\cite{kelk2019finding}.

\begin{corollary}
For a fixed non-negative integer $k$, let $N$ be a level-$k$ network on $X$, and let $A$ be an alignment of characters on $X$. There exists a polynomial $(k+1)$-approximation algorithm to calculate  $PS_\sw(A,N)$ and $PS_{\sw'}(A,N)$.
\end{corollary}

\begin{proof}
Clearly, we can construct a phylogenetic $X$-tree $T$ such that $T\in D(N)$ in time that is polynomial in $|V(N)|$. Furthermore, it takes time that is polynomial in $|X|$ to  calculate $PS(A,T)$ by applying Fitch's algorithm~\cite{fitch1971toward}. The result now follows immediately from Theorem~\ref{t:main}.
\end{proof}

\section{Bounding the rSPR distance} \label{sec:rSPR}

In this section, we establish an upper bound on the rSPR distance between two phylogenetic trees for when both trees are displayed by a given network. Let $N$ be a phylogenetic network, and let $R$ and $R'$ be two switchings of $N$. We define
$$d_\switch(R,R')=h(N)-|R\cap R'|$$
to be the {\it switching distance} between $R$ and $R'$. Intuitively, $d_\switch(R,R')$ is the number of reticulations in $N$ for which $R$ and $R'$ contain different reticulation edges.\\

The following lemma is a generalisation of~\cite[Lemma 3.1]{docker2024hypercubes}.
\begin{lemma}\label{l:rSPR-upper-bound}
Let $N$ be a phylogenetic network on $X$, and let $T_R$ and $T_{R'}$ be two phylogenetic $X$-trees that are yielded by two switchings $R$ and $R'$, respectively, of $N$. Then $$d_\rSPR(T_R,T_{R'})\leq d_\switch(R,R').$$
\end{lemma}

\begin{proof}
Let $S$ (resp. $S'$) be an embedding of $T_R$ (resp. $T_{R'}$) in $N$ whose edge set contains each edge in $R$ (resp. $R'$). Obtain a directed acyclic graph $N'$ from $N$ by deleting each edge that is not contained in $E(S)\cup E(S')$ and, subsequently, applying any of the following two operations until no further operation is possible.
\begin{enumerate}[(i)]
\item Suppress a vertex of in-degree one and out-degree one.
\item If $e$ and $e'$ are two edges in parallel, delete $e'$.
\end{enumerate}
By construction, $N'$ is a phylogenetic network on $X$ and each reticulation edge that is not contained in $R\cup R'$ is deleted in obtaining $N'$. Hence, $h(N')\leq d_\switch(R,R')$. Furthermore, as $S$ and $S'$ are embeddings of $T_R$ and $T_{R'}$, respectively, $T_R$ and $T_{R'}$ are  displayed by $N'$. Now, let $h(T_R,T_{R'})$ be the minimum number of reticulations of any phylogenetic network that displays $T_R$ and $T_{R'}$. Clearly, $h(T_R,T_{R'})\leq h(N')$ and, thus
$$d_\switch(R,R')\geq h(N')\geq h(T_R,T_{R'})\geq d_\rSPR(T_R,T_{R'}),$$ where the last inequality follows from~\cite[Equation 10.1]{Semple:2007ug}. 
\end{proof}

\section{Informative blobs} \label{sec:informativeblobs}
To establish the main result of this paper, we introduce the novel concept of informative and non-informative blobs in this section. After giving a formal definition of these blobs, we establish results related to the changing number of character extensions to embeddings of phylogenetic trees that are displayed by phylogenetic  networks that consist of a single informative or non-informative blob. Let $N$ be a phylogenetic network on $X$, and let $B$ be a blob of $N$. Furthermore, let $C_N(B)$ be the subset of $V(N)$ that contains precisely each vertex that is not in $B$ and that is a child of a vertex of $B$. We refer to $C_N(B)$, as the {\it children} of $B$. As an immediate consequence of the definition of $C_N(B)$, we have the following lemma that we will freely use throughout the remainder of the paper.
\begin{lemma}\label{l:one}
Let $B$ be a blob of a phylogenetic network $N$ on $X$, and let $v$ be a vertex of $C_N(B)$. Furthermore, let $S$ be an embedding of a phylogenetic $X$-tree that is displayed by $N$. Then, $v$ is a vertex of $S$. Moreover, if $v$ is a vertex of a blob $B'$ in $N$, then $v$ is the source of $B'$.
\end{lemma}

\begin{proof}
Suppose that $v$ is a vertex of a blob $B'$ in $N$. Since $v\in C_N(B)$, we have $B'\ne B$. Towards a contradiction, assume that $v$ is not the source of $B'$. It follows that $v$ is a vertex of in-degree two. Let $u$ be the parent of $v$ in $B$, and let $u'$ be the other parent of $v$. Then there is a directed path from the root of $N$ to $v$ that traverses $u$ and a directed path from the root of $N$ to $v$ that traverses $u'$, thereby contradicting that $B$ and $B'$ are two distinct blobs in $N$. We complete the proof by noting that $S$ contains each edge of $N$ whose deletion disconnects $N$ into more than one connected component and, so, $v$ is a vertex of $S$.
\end{proof}

Now, let $s$ be the source of a blob $B$ in a phylogenetic network $N$ on $X$. 
Furthermore, let $S$ be an embedding of a phylogenetic $X$-tree that is displayed by $N$. Let $f$ be a character on $X$, and let $F$ be an extension of $f$ to $V(S)$. We set $\ind(F,B,S)=0$ if each element in $C_N(B)$ is assigned to the same character state under $F$ and, otherwise, we set $\ind(F,B,S)=1$. By Lemma~\ref{l:one}, recall that each vertex in $C_N(B)$ is also a vertex of $S$ and, thus, $\ind(F,B,S)$ is well defined.
Moreover, we say that $B$ is a {\it non-informative blob} relative to $S$ and $f$ if there exists an extension $F$ of $f$ to $V(S)$ such that $PS(f,S)=\ch(F,S)$ and $\ind(F,B,S)=0$. Otherwise, we say that $B$ is an {\it informative blob}. 
We next extend the concept of a single informative blob to all blobs $B_1,B_2,\ldots,B_m$ of $N$ and set $$b(f,N,S)=\min\limits_{\genfrac{}{}{0pt}{1}{F_j\text{ such that}}{ PS(f,S)=\ch(F_j,S)}}\left( \sum_{i=1}^m \ind(F_j,B_i,S)\right),$$ 
where the minimum is taken over all extensions $F_j$ of $f$ to $V(S)$ whose changing number is equal to $PS(f,S)$. Then $b(f,N,S)$ denotes the number of informative blobs relative to $S$ and $f$ in $N$. If $F_j$ is an extension of $f$ to $V(S)$ such that $b(f,N,S)=\sum_{i=1}^m \ind(F_j,B_i,S)$, then we say that $F_j$ {\it realises} $b(f,N,S)$. See Figure~\ref{fig:informative-blob} for an example of a phylogenetic network $N$ on $X = \{x_1, x_2, \ldots, x_8\}$ with two blobs $B$ and $B'$, an embedding $S$ of a phylogenetic $X$-tree displayed by $N$, and a binary character $f$ on $X$ such that $b(f,N,S)=1$. Here, $B$ is non-informative because there exists a minimum extension $F$ of $f$ that assigns character state $0$ to all elements of $C_N(B)$, where $C_N(B)$ contains the source of $B'$ and leaves $x_7$ and $x_8$. Blob $B'$, on the other hand, is informative, as the elements in $C_N(B')=\{x_1, x_2,\ldots, x_6\}$ are assigned two different states by $f$ and thus by any extension of it. To see that $F$ is indeed minimum, notice that $\ch(F,S) = \ch(F,T)=1$. This is minimum since $f$ employs two states, and it is well-known and easy to see that, in this case, any extension of $f$ requires at least one change.

\begin{lemma}\label{l:non-informative}
Let $f$ be a character on $X$, and let $N$ be a phylogenetic network on $X$ with a single blob $B$ whose source $s$ is the child of the root. Let $S$ and $S'$ be embeddings of two phylogenetic $X$-trees that are displayed by $N$. Suppose that $B$ is non-informative relative to $S'$ and $f$. Let $F'$ be an extension of $f$ to $V(S')$ with $\ch(F',S')=PS(f,S')$ that assigns the same character state to each vertex in $C_N(B)$. Then there exists an extension $F$ of $f$ to $V(S)$ such that 
\begin{enumerate}[label={\upshape(\roman*)}]
\item $\ch(F,S)=\ch(F',S')$ and
\item $F(s)=F'(s)$.
\end{enumerate}
\end{lemma}

\begin{proof}
Since $B$ is non-informative, recall that $F'$ exists. Furthermore, by the definition of an embedding, $s$ is the child of the roots of $S$ and $S'$. Let $V$ be the subset of $V(N)$ that precisely contains each vertex that is not a vertex of $B$. Since $B$ is the only blob of $N$, each vertex in $V$ is also a vertex of $S$ and $S'$. Furthermore, each vertex of $S$ or $S'$ that is not in $V$, is a vertex of $B$. 
Since $s$ is the child of the root of $S'$ and $F'$ assigns the same character state, say $\alpha$, to each vertex in $C_N(B)$, it follows that $F'$ also assigns $\alpha$ to each vertex in $V(S')$ that is an ancestor of some vertex in $C_N(B)$. In particular, $F'(s)=\alpha$.

Now, consider $S$. Set $F(u)=F'(u)$ for each vertex $u\in V$ and set $F(u')=\alpha$ for each vertex $u'\in V(S)\setminus V$. By definition of $F$, we again have that each vertex in $V(S)$ that is an ancestor of some vertex in $C_N(B)$ is assigned to $\alpha$ under $F$. Hence, as $F'$ is an extension of $f$ to $V(S')$, $F$ is an extension of $f$ to $V(S)$ with $F(s)=F'(s)=\alpha$; thereby satisfying (ii). 
Moreover, since $S$ and $S'$ are embeddings of two phylogenetic $X$-trees that are displayed by $N$, the edges of $N$ satisfy the following property: If $e=(u,v)$ is an edge of $S'$ (resp. $S$) but not an edge of $S$ (resp. $S'$), then $e$ is an edge of $B$ and, consequently, $F'(u)=F'(v)=\alpha$ (resp. $F(u)=F(v)=\alpha$).
It follows that $\ch(F,S)=\ch(F',S')$ which satisfies (i) and, therefore, completes the proof of the lemma.
\end{proof}

\begin{lemma}\label{l:changing-number-rSPR}
Let $f$ be a character on $X$. Let $T$ and $T'$ be two phylogenetic $X$-trees. 
Furthermore, let $F'$ be an extension of $f$ to $V(T')$. 
Then there exists an extension $F$ of $f$ to $V(T)$ such that
\begin{enumerate}[label={\upshape(\roman*)}]
\item $\ch(F,T) \leq \ch(F',T') + d_\rSPR(T',T)$ and
\item $F(\rho)=F'(\rho')$, where $\rho$ and $\rho'$ is the root of $T$ and $T'$, respectively.
\end{enumerate}
\end{lemma}

\begin{proof}
We show by induction on $d_\rSPR(T',T)$ that there exists an extension $F$ of $f$ to $V(T)$ that satisfies (i)--(ii). Suppose that  $d_\rSPR(T',T) = 1$. Then there exists a single rSPR operation that transforms $T'$ into $T$. Given such an rSPR operation, let $(u',v')$ be the edge of $T'$ that is deleted in the pruning part of the operation. Let $u_p'$ and $u_c' \neq v'$ be the parent and other child of $u'$ in $T'$. Further, let $u$ be the vertex that subdivides an edge, say $(u_p,u_c)$, when reattaching the resulting subtree with root $v'$ such that $(u,v')$ is an edge in $T$. 
Noting that each vertex in $T$ except for $u$ is also a vertex of $T'$, we next obtain an extension $F$ of $f$ to $V(T)$ with no character state transition on the root edge of $T$ as follows: For each vertex $w \neq u$, we set $F(w) = F'(w)$. 
In particular, we have $F(\rho)=F'(\rho')$ and, so, (ii) follows. Moreover, if $u_p=\rho$, we set $F(u) = F(u_p)$. Otherwise, we set $F(u) = \alpha$, where $\alpha$ is a character state that has been assigned to at least one neighbour of $u$ in $T$ under $F$ and there is no other character state that has been assigned to strictly more neighbours of $u$ in $T$ under $F$.  We next show that (i) holds. Consider the edges of $T$. Except for the edges $(u,v')$ used to reattach the subtree with root $v'$, $(u_p',u_c')$ obtained from suppressing $u'$, and $(u_p,u)$ and $(u,u_c)$ obtained from subdividing $(u_p,u_c)$ with $u$, each edge of $T$ is also an edge of $T'$. If  $F(u_p')\ne F(u_c')$, then either $F'(u_p')\ne F'(u')$ or $F'(u')\ne F'(u_c')$. Hence, suppressing $u'$ does not increase the changing number. On the other hand, when assigning a character state to $u$ the changing number may increase. More specifically, we consider three cases. 
First, if $F(u_p)=F(u_c)$, then $F(u)=F(u_p)$ by definition of $F$. Note that $u_p$ may be $\rho$. Thus, there is no character state transition on the two edges $(u_p,u)$ and  $(u,u_c)$, and at most one such transition on the edge $(u,v')$ under $F$ in $T$.  Second, if $|\{F(u_p), F(u_c),F(v')\}|=3$, then there is a character state transition on the edge $(u_p, u_c)$ under $F'$ and we have two character state transitions on the three edges $(u_p,u)$, $(u,u_c)$, and $(u,v')$ under $F$. Third, if $F(u_p)\ne F(u_c)$ and $|\{F(u_p), F(u_c),F(v')\}|=2$, then there is again a character state transition on the edge $(u_p, u_c)$ under $F'$ and we have one character state transition on the three edges $(u_p,u)$, $(u,u_c)$, and $(u,v')$ under $F$.
Hence, regardless of which case applies
\[
\ch(F,T) \leq \ch(F', T') + 1 =  \ch(F', T') + d_\rSPR(T',T);
\]
thereby satisfying (i) for when $d_\rSPR(T',T)=1$. 

Now suppose that $d_\rSPR(T',T) \geq 2$ and that (i)--(ii) are satisfied for all pairs of phylogenetic trees whose rSPR distance is strictly smaller than $ d_\rSPR(T',T)$. Let $T''$ be a phylogenetic $X$-tree  such that $ d_\rSPR(T', T'') = 1 $ and $
d_\rSPR(T'', T) =d_\rSPR(T',T)-1$. Recalling that $F'$ is an extension of $f$ to $V(T')$, it follows from the induction hypothesis, that there is an extension $F''$ of $f$ to $V(T'')$ that satisfies (ii) and  $\ch(F'',T'') \leq \ch(F',T') + 1$. Again, by the induction hypothesis, there exists an extension $F$ of $f$ to $V(T)$ that satisfies (ii) and $\ch(F,T) \leq \ch(F'',T'') + d_\rSPR(T'', T)$. Hence, by combining the two inequalities we obtain 
\begin{align*}
\ch(F,T)&\leq \ch(F'',T'') + d_\rSPR(T'', T)\\
&\leq  \ch(F',T') + 1 + d_\rSPR(T'', T) = \ch(F',T') + d_\rSPR(T',T)
\end{align*}
and $F(\rho)=F'(\rho')$. Hence, $F$ satisfies (i)--(ii). This completes the proof of the lemma.
\end{proof}

\begin{corollary}\label{c:changing-number-rSPR}
Let $f$ be a character on $X$. Let $N$ be a phylogenetic network on $X$ with a single blob, and let $T$ and $T'$ be two phylogenetic $X$-trees displayed by $N$. Furthermore, let $F'$ be an extension of $f$ to $V(T')$. Then there exists an extension $F$ of $f$ to $V(T)$ such that
\begin{enumerate}[label={\upshape(\roman*)}]
\item $\ch(F,T) \leq \ch(F',T') + k$, where $k$ is the level of $N$, and
\item $F(\rho)=F'(\rho')$, where $\rho$ and $\rho'$ is the root of $T$ and $T'$, respectively.
\end{enumerate}
\end{corollary}

\begin{proof}
Let $R$ and $R'$ be two switchings of $N$ that yield $T$ and $T'$, respectively. By Lemma~\ref{l:rSPR-upper-bound}, we have $d_\rSPR(T,T') \leq d_\switch(R, R')$. Noting that $N$ has a single blob, we have $d_\switch(R, R') \leq k$ and the corollary now follows from Lemma~\ref{l:changing-number-rSPR}. 
\end{proof}

While Lemma~\ref{l:non-informative} is restricted to phylogenetic networks that consist of a single non-informative blob, the next lemma establishes an analogous result for all phylogenetic networks that consist of a single blob.

\begin{lemma}\label{l:informative}
Let $f$ be a character on $X$, and let $N$ be a phylogenetic network on $X$ with a single blob $B$ whose source $s$ is the child of the root. Let $S$ and $S'$ be embeddings of two phylogenetic trees that are displayed by $N$. Furthermore, let $F'$ be an extension of $f$ to $V(S')$.
Then there exists an extension $F$ of $f$ to $V(S)$ such that 
\begin{enumerate}[label={\upshape(\roman*)}]
\item $\ch(F,S) \leq \ch(F',S') + k$, where $k$ is the level of $N$, and 
\item $F(s) = F'(s)$.
\end{enumerate}
\end{lemma}

\begin{proof}
Let $T$  (resp. $T'$) be the two phylogenetic $X$-trees such that $S$ (resp. $S'$) is an embedding of $T$ (resp. $T'$) in $N$. 
Furthermore, let $F'(s)=\alpha$. Observe that each vertex in $T'$ is a unique degree-three vertex in $S'$. First, let $F'_{T'}$ be the extension of $f$ to $V(T')$ obtained from $F'$ by setting $F'_{T'}(w) = F'(w)$ for each $w \in V(T')$.
Then
\begin{equation}\label{eq:a}
\ch(F'_{T'}, T') \leq \ch(F', S'),
\end{equation}
and, because there is no character state transition on any edge of the root path of $S'$ that contains $s$, the root of $T'$ is assigned to $\alpha$.
Second, by Corollary~\ref{c:changing-number-rSPR}, there exists an extension $F_T$ of $f$ to $V(T)$ such that 
\begin{equation}\label{eq:b}
\ch(F_{T}, T) \leq \ch(F'_{T'},T') + k
\end{equation}
and the root of $T$ is also assigned to $\alpha$.
Third, we obtain an extension $F$ of $f$ to $V(S)$ from $F_T$ as follows. Noting that each edge $(v,v')$ in $T$ corresponds to a unique directed path $v=v_1,v_2,\ldots,v_s=v'$ in $S$ whose non-terminal vertices all have degree two, we set $F(v_1)=F(v_2)=\cdots=F(v_{s-1})=F_T(v)$ and $F(v_s)=F_T(v')$. 
Then
\begin{equation}\label{eq:c}
\ch(F, S)=\ch(F_{T}, T),
\end{equation}
and, since the root of $T$ is assigned to $\alpha$, it follows that each vertex on the root path of $S$, in particular $s$, is also assigned to $\alpha$. Hence (ii) is satisfied.
Moreover, by combining Equations~(\ref{eq:a})--(\ref{eq:c}), we have
\[
\ch(F,S) = \ch(F_{T}, T) \leq \ch(F'_{T'},T') + k \leq \ch(F', S') + k.
\]
This concludes the proof of the lemma.
\end{proof}

\section{Blob reduction} 
\label{sec:clusterreduction}
In this section, we introduce the notion of a blob reduction. Intuitively, this allows us to decompose a phylogenetic network $N$ into two smaller phylogenetic networks and calculate the parsimony score of an embedding $S$ of a phylogenetic $X$-tree displayed by $N$ based on these two smaller networks.

Let $N$ be a phylogenetic network on $X$, let $B$ be a blob of $N$ whose source, say $s$, has maximum distance from the root of $N$, and let $Y=\cl(s)$. For some $y\notin X$, the {\it blob reduction} of $B$ reduces $N$ to two smaller phylogenetic networks as follows. Let $N(\bar Y)$ be the phylogenetic network on $(X\setminus Y)\cup\{y\}$ that is obtained from $N$ by replacing the subnetwork of $N$ that is rooted at $s$ with a single new leaf $y$. Furthermore, let $N(Y)$ be the phylogenetic network on $Y$ that is obtained from the subnetwork of $N$ that is rooted at $s$ by adding a new vertex $\rho_Y$ and edge $(\rho_Y,s)$. 
By construction, each of $N(\bar Y)$ and $N(Y)$ contains at least one leaf.

Now, let $S$ be an embedding of a phylogenetic $X$-tree $T$ that is displayed by $N$. Recall that $s$ is a vertex of $S$ and $\cl(s)=Y$. 
Let $f$ be a character on $X$, and let $F$ be an extension of $f$ to $V(S)$. Using the aforementioned blob reduction of $B$ as a guide, we next also reduce $S$ to two smaller trees such that one of the resulting trees is an embedding of a subtree of $T$ in $N(\bar Y)$ and the other one is an embedding of another subtree of $T$ in $N(Y)$. More specifically, let $S(\bar Y)$ be the tree with leaf set $(X\setminus Y)\cup\{y\}$ that is obtained from $S$ by replacing the subtree of $S$ that is rooted at $s$ with a single new leaf $y$. Furthermore, let $S(Y)$ be the tree with leaf set $Y$ that is obtained from the subtree of $S$ that is rooted at $s$ by adding a new vertex $\rho_Y$ and edge $(\rho_Y,s)$. We call $(S(Y), S(\bar{Y}))$ the {\it cluster tree pair} of $S$ relative to $B$.
Let $f_Y$ and $f_{\bar{Y}}$ be a character on $Y$ and on $(X\setminus Y)\cup\{y\}$, respectively, such that $f_Y(\ell)=f(\ell)$ for each $\ell\in Y$ and $f_{\bar{Y}}(\ell')=f(\ell')$ for each $\ell'\in X\setminus Y$. We refer to extensions $F_Y$ of $f_Y$ to $V(S(Y))$ and $F_{\bar{Y}}$ of $f_{\bar Y}$ to $V(S(\bar Y))$ as a {\it pair of cluster extensions with respect to $f$} if $F_{\bar Y}(y)=F_Y(\rho_Y)$. 
Except for $f_{\bar{Y}}(y)$, observe that $f$ uniquely determines the character state of each leaf in $S(Y)$ and $S(\bar{Y})$. Moreover, since the root path of $S(Y)$ contains at least one edge, the definition of a pair of cluster extensions implies that $F_Y(\rho_Y)=F_Y(s)$ by our assumption following Observation~\ref{ob:root-edge}. The next lemma shows how the changing number of extensions of  characters $f$, $f_Y$, and $f_{\bar Y}$ to $V(S)$,  $V(S(Y))$, and $V(S(\bar Y))$, respectively, are related to each other.

\begin{lemma}\label{l:cluster-one}
Let $B$ be a blob of a phylogenetic network $N$ on $X$ to which the blob reduction can be applied. Let $f$ be a character on $X$, and let $S$ be an embedding of a phylogenetic $X$-tree that is displayed by $N$. Furthermore, let $(S(Y), S(\bar Y))$ be the cluster tree pair of $S$ relative to $B$. 
Then, the following two statements hold.
\begin{enumerate}[label={\upshape(\roman*)}]
   \item If $F$ is an extension of $f$ to $V(S)$, then there exists a pair of cluster extensions $(F_Y, F_{\bar Y})$ such that  \[\ch(F,S) = \ch(F_Y, S(Y)) + \ch(F_{\bar{Y}}, S(\bar{Y})).\]
   \item If $(F_Y, F_{\bar Y})$ is a pair of cluster extensions with respect to $f$, then there exists an extension $F$ of $f$ to $V(S)$ such that \[\ch(F, S) = \ch(F_Y, S(Y)) + \ch(F_{\bar{Y}}, S(\bar{Y})).\] 
\end{enumerate}
\end{lemma}

\begin{proof}
Let $s$ be the source of $B$, and let $Y=\cl_N(s)$. By the definition of a cluster tree pair, $S(Y)$ has leaf set $Y$ and root $\rho_Y$, and $S(\bar Y)$ has leaf set $(X\setminus Y)\cup\{y\}$ and root $\rho$, where $\rho$ is also the root of $S$. Lastly, as $s$ is a vertex of $S$, it follows from the construction of $S(Y)$ and $S(\bar Y)$ that $s$ corresponds to the child of $\rho_Y$ in $S(Y)$ and  to $y$ in $S(\bar Y )$, whereas each other vertex of $S$ corresponds to a unique vertex in either $S(Y)$ or $S(\bar Y)$. To ease reading, we refer to the child of $\rho_Y$ in $S(Y)$ as $s_Y$. Reversely, the only vertex of $S(Y)$ and $S(\bar Y)$ that does not correspond to a unique vertex in $S$ is $\rho_Y$. We next show that (i) and (ii) hold.

First, let $F$ be an extension of $f$ to $V(S)$. Obtain a pair of cluster extensions $F_Y$ and $F_{\bar{Y}}$ of characters $f_Y$ and $f_{\bar{Y}}$ to $V(S(Y))$ and $V(S(\bar{Y}))$, respectively, in the following way. For each vertex $w$ of $V(S(Y))\setminus \{\rho_Y\}$, set $F_Y(w)=F(w')$, where $w'$ is the vertex of $S$ that $w$ corresponds to, and set $F_Y(\rho_Y)=F_Y(s_Y)$. Similarly, for each vertex $w$ of $V(S(\bar{Y}))$, set $F_{\bar{Y}}(w)=F(w')$, where $w'$ is the vertex of $S$ that $w$ corresponds to. Since $y$ and $s_Y$ both correspond to $s$, it follows that  $F_{\bar Y}(y)=F_Y(\rho_Y)$. It is is now easily checked that $F_Y$ and $F_{\bar{Y}}$ is a pair of cluster extensions with respect to $f$ and that $$\ch(F,S)= \ch(F_Y,S(Y))+\ch(F_{\bar{Y}},S(\bar{Y})).$$ Hence, (i) holds.

Now, let $F_Y$ and $F_{\bar{Y}}$ be a pair of cluster extensions with respect to $f$. In particular, $F_Y(\ell)=f(\ell)$ for each $\ell\in Y$, $F_{\bar{Y}}(\ell)=f(\ell)$ for each $\ell\in X\setminus Y$, and $F_{\bar Y}(y)=F_Y(\rho_Y)$. Now obtain an extension $F$ of $f$ to $V(S)$ from $F_Y$ and $F_{\bar{Y}}$ in the following way. For each vertex $w'$ of $V(S)$ that corresponds to a vertex $w$ of $S(Y)$, set $F(w')=F_Y(w)$ and, for each vertex $w'$of $V(S)\setminus \{s\}$ that corresponds to a vertex $w$ of $S(\bar{Y})$, set $F(w')=F_{\bar{Y}}(w)$.
Since $F_Y(s_y)=F_Y(\rho_Y)$, it follows that  $$\ch(F,S) = \ch(F_Y,S(Y))+\ch(F_{\bar{Y}},S(\bar{Y})).$$ Thus, (ii) holds as well.
\end{proof}

\section{Proof of Theorem~\ref{t:main}}\label{sec:main}
In this section, we establish the proof of Theorem~\ref{t:main} and show that the bound that is given in the theorem is sharp. 
Most work in proving Theorem~\ref{t:main} goes into establishing the following lemma. 

\begin{lemma}\label{l:main}
Let $f$ be a character on $X$. Let $N$ be a phylogenetic network on $X$, and let $S$ and $S'$ be embeddings of two phylogenetic $X$-trees that are displayed by $N$. Furthermore, let $F'$ be an extension of $f$ to $V(S')$ that realises $b(f,N,S')$. 
Then there exists an extension $F$ of $f$ to $V(S)$ such that 
$$\ch(F,S) \leq \ch(F',S') + k\cdot b(f,N,S'),$$
where $k$ is the level of $N$.
\end{lemma}

\begin{proof} 
Let $B_1, B_2, \ldots, B_m$ be the blobs of $N$. The proof is by induction on $m$. If $m=0$, then $N$ is a phylogenetic tree with $k=0$ and so the result clearly follows since $S=S'$ and, therefore, 
\[
\ch(F,S)\leq \ch(F',S')+0
\]
when setting $F = F'$. Now assume that $m\geq 1$ and that the statement is true for all phylogenetic networks with at most $m-1$ blobs. Let $B$ be a blob of $N$ whose source $s$ has maximum distance from the root of $N$ over all its blobs, and let $Y = \cl_N(s)$. Without loss of generality, we assume that $B=B_m$.

For some $y\notin X$, let $N(\bar Y)$ be the phylogenetic network on $(X \setminus Y) \cup \{y\}$, and let $N(Y)$ be the phylogenetic network on $Y$ and with root $\rho_Y$ resulting from $N$ by applying a blob reduction to $B_m$.
Notice that by construction, $N(Y)$ consists of the single blob $B_m$ with $s$ being the child of  $\rho_Y$, whereas $N(\bar Y)$ contains precisely $m-1$ blobs. Moreover, let $(S(Y), S(\bar Y))$ be the cluster tree pair of $S$ relative to $B_m$, and let $(S'(Y), S'(\bar Y))$ be the cluster tree pair of $S'$ relative to $B_m$. Since $s$ is a vertex of $S$ and $S'$, $s$ is also a vertex of $S(Y)$ and $S'(Y)$.  Now, by Lemma~\ref{l:cluster-one}, Part (i), there exists a pair of cluster extensions $(F'_Y, G'_{\bar Y})$ with respect to $f$ such that
\begin{equation}\label{eq:cluster-reduction-optimal-ext}
\ch(F',S') = \ch(F'_Y, S'(Y)) + \ch(G'_{\bar{Y}}, S'(\bar{Y}))
\end{equation}
and $G'_{\bar Y}(y) = F'_Y(\rho_Y) = F'_Y(s)$.  Let $f_Y$ and $g_{\bar Y}$ be the characters on $Y$ and $(X\setminus Y)\cup\{y\}$, respectively, such that $F'_Y$ and $G'_{\bar{Y}}$ are extensions of $f_Y$ and $g_{\bar{Y}}$, respectively. 

We next consider $N(\bar{Y})$ and start by making two observations.
First, since $\ch(F',S') = PS(f,S')$, it follows from Lemma~\ref{l:cluster-one}, Parts (i) and (ii), that $\ch(G'_{\bar{Y}}, S'(\bar{Y})) = PS(g_{\bar Y},S'(\bar Y))$. Second, by the construction of the pair of cluster extensions in Lemma~\ref{l:cluster-one} Part (i), we may assume that, for each vertex $w$ of $V(S'(\bar{Y}))$, we have $G'_{\bar Y}(w) = F'(w')$, where $w'$ is the vertex of $S'$ that $w$ corresponds to. Then, as $F'$ realises $b(f,N,S')$, $G'_{\bar Y}$ realises $b(g_{\bar{Y}},N(\bar{Y}),S'(\bar{Y}))$.
Noting that $N(\bar{Y})$ has $m-1$ blobs and level at most $k$, we now apply the induction hypothesis to obtain an extension $G_{\bar{Y}}$ of $g_{\bar Y}$ to $V(S(\bar{Y}))$ that satisfies
\begin{equation}\label{eq:induction_upper_part}
\ch(G_{\bar{Y}},S(\bar{Y})) \leq \ch(G'_{\bar{Y}},S'(\bar{Y})) + k\cdot b(g_{\bar{Y}},N(\bar{Y}),S'(\bar{Y}))
\end{equation}
such that $G_{\bar Y} (y) = g_{\bar Y}(y) = G'_{\bar Y} (y) = F'_Y(s)$. 

To complete the proof, we consider $N(Y)$. Here, we distinguish two cases depending on whether its single blob $B_m$ whose level is at most $k$ is informative or non-informative.

First, assume that $B_m$ is informative. By Lemma~\ref{l:informative}, there exists an extension $F_Y$ of $f_Y$ to $V(S(Y))$ such that
\begin{equation}\label{eq:lower_part_informative}
    \ch(F_Y,S(Y)) \leq \ch(F'_Y,S'(Y)) +k \quad \text{and} \quad F_Y(s) = F'_Y(s).
\end{equation}
Since $F_Y(s) = F'_Y(s) = G_{\bar Y}(y)$, the pair $(F_Y, G_{\bar Y})$ is a pair of cluster extensions with respect to $f$. Thus, by Lemma~\ref{l:cluster-one}, Part (ii), there exists an extension $F$ of $f$ to $V(S)$ such that
\[ \ch(F,S) = \ch(F_Y, S(Y)) + \ch(G_{\bar Y}, S(\bar Y)).\]
Now, using Inequalities \eqref{eq:induction_upper_part} and \eqref{eq:lower_part_informative}, we obtain
\begin{align*}
    \ch(F,S) &= \ch(F_Y, S(Y)) + \ch(G_{\bar Y}, S(\bar Y)) \\
    &\leq  \ch(F'_Y, S'(Y)) +k +  \ch(G'_{\bar{Y}},S'(\bar{Y})) + k\cdot b(g_{\bar{Y}},N(\bar{Y}),S'(\bar{Y})) \\
    &= \ch(F'_Y, S'(Y)) + \ch(G'_{\bar{Y}},S'(\bar{Y})) + k \cdot (1+b(g_{\bar{Y}},N(\bar{Y}),S'(\bar{Y}))) \\
    &= \ch(F',S') + k \cdot b(f,N,S'),
\end{align*}
where the last equality follows from Equation~\eqref{eq:cluster-reduction-optimal-ext} and the fact that $B_m$ is informative.

Second, assume that $B_m$ is non-informative. Then by Lemma~\ref{l:non-informative}, there exists an extension $F_Y$ of $f_Y$ to $V(S(Y))$ such that
\begin{equation}\label{eq:lower_part_noninformative}
\ch(F_Y,S(Y)) = \ch(F'_Y,S'(Y)) \quad \text{and} \quad F_Y(s) = F'_Y(s). 
\end{equation}
The remainder of the proof is now similar to the first case. In particular, noting that the pair $(F_Y, G_{\bar Y})$ is a pair of cluster extensions with respect to $f$, by Lemma~\ref{l:cluster-one}, Part(ii), there exists an extension $F$ of $f$ to $V(S)$ such that
\begin{align*}
     \ch(F,S) &= \ch(F_Y, S(Y)) + \ch(G_{\bar Y}, S(\bar Y)).
\end{align*}
Using Inequalities in \eqref{eq:induction_upper_part} and \eqref{eq:lower_part_noninformative}, we obtain
\begin{align*}
    \ch(F,S) &= \ch(F_Y, S(Y)) + \ch(G_{\bar Y}, S(\bar Y)) \\
    &\leq  \ch(F'_Y, S'(Y))  +  \ch(G'_{\bar{Y}},S'(\bar{Y})) + k\cdot b(g_{\bar{Y}},N(\bar{Y}),S'(\bar{Y})) \\
    &= \ch(F',S') + k \cdot b(f,N,S'),
\end{align*}
where the last equality follows from Equation~\eqref{eq:cluster-reduction-optimal-ext} and the fact that $B_m$ is non-informative.

In both cases, we obtain an extension $F$ of $f$ to $V(S)$ such that 
\[
\ch(F,S) \leq \ch(F',S') + k\cdot b(f,N,S').
\]
This concludes the proof of the lemma.
\end{proof}

\begin{corollary}\label{c:main}
Let $f$ be a character on $X$. Let $N$ be a phylogenetic network on $X$, and let $S$ and $S'$ be embeddings of two phylogenetic $X$-trees displayed by $N$ such that $PS_{\sw}(f,N)=PS(f,S')$. Then $$PS(f,S)\leq PS(f,S')+k\cdot b(f,N,S'),$$ where $k$ is the level of $N$.
\end{corollary}

\begin{proof}
Let $F'$ be an extension of $f$ to $V(S')$ that 
realises $b(f,N,S')$. 
By Lemma~\ref{l:main}, there exists an extension $F$ of $f$ to $V(S)$ such that 
$PS(f, S) \leq \ch(F,S) \leq PS(f, S') + k\cdot b(f,N,S')$.
\end{proof}

We are finally in a position to establish the main result of this paper, which we restate for convenience.\\

\noindent{\bf Theorem~\ref{t:main}.} {\it Let $N$ be a phylogenetic network on $X$, and let $T$ be a phylogenetic $X$-tree in $D(N)$. Let $A$ be an alignment of characters on $X$. Then 
$$PS(A,T)\leq (k+1)\cdot PS_{\sw}(A,N) \text{ and }PS(A,T)\leq (k+1)\cdot PS_{\sw'}(A,N),$$ 
where $k$ is the level of $N$. }\\

\begin{proof}
We establish that $$PS(A,T)\leq (k+1)\cdot PS_{\sw}(A,N).$$ Since $PS_{\sw}(A,N)\leq PS_{\sw'}(A,N)$, it immediately follows that $PS(A,T)\leq (k+1)\cdot PS_{\sw'}(A,N)$ also holds.

First, assume that $A$ consists of a single character $f$. 
Let $S$ be an embedding of $T$ in $N$, and let $S'$ be an embedding of a phylogenetic $X$-tree in $N$ such that $PS_{\sw}(f, N) =  PS(f, S')$. 
By Corollary~\ref{c:main}, we have
\begin{align*}
    PS(f,T)=PS(f, S) \leq PS(f, S') + k \cdot b(f, N,S'),
\end{align*}
where the first equality follows from Lemma~\ref{l:fitch}. Now, as every blob $B$ in $N$ that is informative relative to $S'$ contributes at least one to $PS(f, S')$, we have
\begin{align*}
    b(f, N,S') \leq PS(f, S').
\end{align*}
By Lemma~\ref{l:fitch} and combining the last two inequalities, we obtain
\begin{align}
    PS(f,T) = PS(f,S) &\leq PS(f,S') + k \cdot PS(f, S') \nonumber\\
    &= (k+1) \cdot PS(f, S') \nonumber\\
    &= (k+1) \cdot PS_{\sw}(f, N). 
   \label{eq:single} 
\end{align}

Now, assume that $A = (f_1, \ldots, f_n)$ with $n \geq1$. Then, we apply Inequality~\eqref{eq:single} to each character, and obtain
\begin{align*}
    PS(A, T) &= \sum\limits_{i=1}^n PS(f_i, T) \\
    &\leq \sum\limits_{i=1}^n (k+1) \cdot PS_{\sw}(f_i, N) \\
    &\leq (k+1) \cdot \sum\limits_{i=1}^n PS_{\sw}(f_i, N) \\
    &= (k+1) \cdot PS_{\sw}(A,N).
\end{align*}
\end{proof}

\begin{figure}[t]
    \centering
    \includegraphics[width=\textwidth]{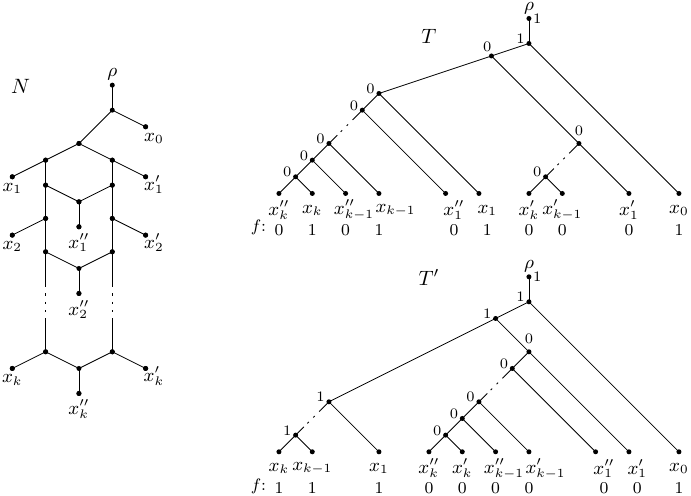}
    \caption{A level-$k$ network $N$ on $X = \{x_0\} \cup \{x_i, x'_i, x''_i : 1 \leq i \leq k\}$ (left) and two phylogenetic $X$-trees $T$ and $T'$ displayed by $N$ together with a binary character $f$ and extensions $F$ and $F'$ to $V(T)$ and $V(T')$, respectively (right).}   
    \label{fig:tight-example}
\end{figure}

We close this section by presenting, for each $k \geq 0$, a level-$k$ network and a binary character such that the upper bound stated in Theorem~\ref{t:main} is sharp. As level-0 networks on $X$ are phylogenetic $X$-trees, the bound is sharp for any level-0 network and any binary character on $X$. For $k \geq 1$, let $N$ be the level-$k$ network on $X = \{x_0\} \cup \{x_i, x'_i, x''_i: 1 \leq i \leq k\}$ that is depicted in Figure~\ref{fig:tight-example}. Further, let $f\colon X \rightarrow \{0,1\}$ be the binary character with $f(x) = 1$ if $x \in \{x_0, x_1, \ldots, x_k\}$ and $f(x) = 0$ otherwise. Let us consider the two phylogenetic $X$-trees $T, T'\in D(N)$ that are illustrated on the right-hand side of Figure~\ref{fig:tight-example} together with extensions $F$ and $F'$ of $f$ to $V(T)$ and $V(T')$, respectively.
By using Fitch's algorithm it is easy to verify that $F$ and $F'$ are minimum. Hence, we have $PS(f,T') = 1$ and, thus, $PS_{\sw}(f,N) = 1$ (since $f$ employs two character states, $PS_{\sw}(f,N) \geq 1$). 
Moreover, $PS(f,T) = k+1$. In summary, we have $PS(f,T) = (k+1)\cdot PS_{\sw}(f,N)$. As the construction shown in Figure~\ref{fig:tight-example} involves a single character, the two notions of softwired parsimony on $N$ coincide, and we have $PS(f,T) = (k+1)\cdot PS_{\sw'}(f,N)$ for the same example.

\section{Parental parsimony for phylogenetic networks} \label{sec:parental}
We now briefly consider the notion of \emph{parental parsimony} introduced by~\citet{van2018improved} as an alternative to softwired and hardwired parsimony. Intuitively, instead of defining the parsimony score of a phylogenetic network $N$ by considering its display set $D(N)$, parental parsimony considers the set of \emph{parental trees} (sometimes also called \emph{weakly displayed trees}~\cite{huber2016weakly}), which is a superset of $D(N)$. 

A {\it multilabelled tree} on $X$ is a leaf-labelled rooted tree  whose root has out-degree one, all other interior vertices have in-degree one and out-degree one, or in-degree one and out-degree two, and, for each element $x$ in $X$, there exists at least one leaf in $T$ that is labelled $x$.
Now using the same notation as \cite{van2018improved}, let $U^\ast(N)$ be the multilabelled obtained from a  phylogenetic network $N$ on $X$ as follows: The vertices of $U^\ast(N)$ are the directed paths in $N$ starting at the root of $N$, and for each pair of directed paths $p, p'$, there is an edge $(p,p')$ in $U^\ast(N)$  if and only if $p'$ is an extension of $p$ by one additional edge of $N$. 
Furthermore, each vertex in $U^\ast(N)$ corresponding to a path in $N$ starting at the root of $N$ and ending at $x \in X$ is labelled by $x$. For an example of the multilabelled tree $U^\ast(N)$ obtained from a phylogenetic network $N$, see Figure~\ref{fig:parental-tree}. Now, a phylogenetic $X$-tree is called a \emph{parental tree} of $N$ if it can be obtained from a subgraph of $U^\ast(N)$ by suppressing vertices of in-degree and out-degree one.
To denote the set of all parental trees of $N$, we use $P(N)$. Informally speaking, a tree is a parental tree of a phylogenetic network if it can be drawn inside the network in such a way that the tree vertices of the tree correspond to tree vertices of the network. Importantly, though, a parental tree is not necessarily a displayed tree (see Figure~\ref{fig:parental-tree}), whereas every displayed tree is also parental. 

\begin{figure}[t]
    \centering
    \includegraphics[width=\textwidth]{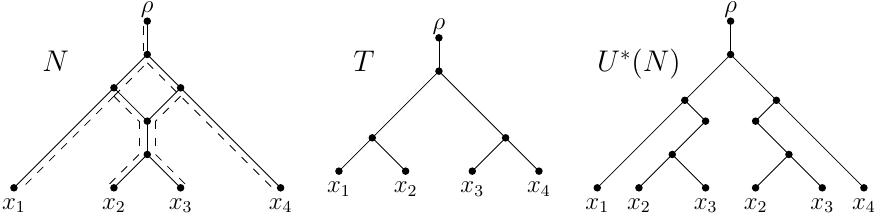}
    \caption{A phylogenetic network $N$ on $X=\{x_1, x_2, \ldots, x_4\}$, a phylogenetic $X$-tree $T$ such that $T\in P(N)$ and $T\notin D(N)$, and the multilabelled tree $U^\ast(N)$. The dashed lines in $N$ indicate how $T$ can be drawn inside $N$.}
    \label{fig:parental-tree}
\end{figure}

Given a character $f$ on $X$ and a phylogenetic network $N$ on $X$, the \emph{parental parsimony score} of $f$ on $N$ is now defined as
\begin{align*}
    PS_{\pa}(f,N) &= \min\limits_{T \in P(N)} PS(f,T), 
\end{align*}
where the minimum is taken over all parental trees for $N$. 

It was shown by \citet[Theorem 2]{van2018improved} that computing the parental parsimony score is NP-hard even if $f$ is a binary character and $N$ is a restricted type of a so-called tree-child network~\cite{cardona2009comparison}. 
It is thus a natural question if our main result for softwired parsimony (Theorem~\ref{t:main}) generalises to parental parsimony. Unfortunately, this is not the case. Suppose that $n$
is an even integer and that $N$ is the level-$1$ network on $n+1$ leaves depicted in Figure~\ref{fig:parental-counterexample}. Then, both phylogenetic trees $T$ and $T'$ as depicted in the same figure are parental trees of $N$ and $T'\notin D(N)$. Additionally, suppose that $f$ is the binary character that assigns state $0$ to leaves $x_2, x_4, x_6, \ldots, x_{n}$, and state $1$ to leaves $x_1, x_3, \ldots, x_{n+1}$.  Crucially, we have $PS(f, T) = n/2$ and $PS(f,T')=1$.
In particular, $PS_{\pa}(f,N)=1$, and thus, $$PS(f, T) = n/2 \not\leq 2 \cdot PS_{\pa}(f,N)= 2$$ for each $n > 4$ (a similar argument applies to $n$ being odd), which shows that Theorem~\ref{t:main} does not generalise to parental parsimony even if the phylogenetic network is level-$1$ with a single blob and the alignment consists of a single binary character.

\begin{figure}[t]
    \centering
    \includegraphics[width=.9\textwidth]{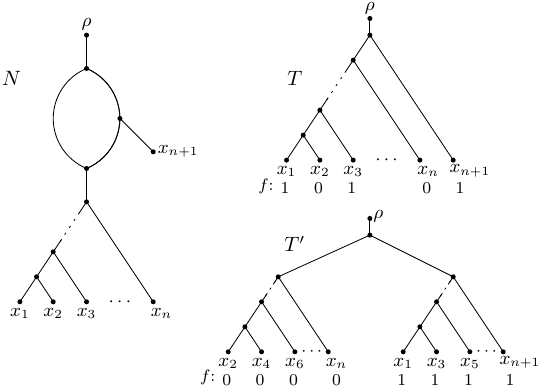}
    \caption{A level-$1$ network $N$ on $X=\{x_1,x_2, \ldots, x_{n+1}\}$ with $n$ being even and two parental trees $T$ and $T'$ of $N$ together with a binary character $f$ on $X$. Here, $PS_{\pa}(N,f) = PS(f,T') = 1$, whereas $PS(f,T)=n/2$.}
    \label{fig:parental-counterexample}
\end{figure}

\section{Softwired parsimony for semi-directed and unrooted networks}\label{sec:semi-directed}
In this section, we show that binary semi-directed networks have an analogous bound on the softwired parsimony score as the one we established for rooted binary phylogenetic networks. We briefly turn our attention to unrooted binary phylogenetic networks at the end of the section and show that the approach we take to establish the bound for semi-directed networks does not work in this setting. However, this does not exclude the possibility that similar bounds can be obtained for unrooted phylogenetic networks by other means.

A {\it binary semi-directed network} $N_s$ on $X$ is a leaf-labelled mixed multigraph without any loops that can be obtained from a rooted binary phylogenetic network $N_r$ by deleting its root, suppressing the child of the root, and omitting the direction of each edge that is not a reticulation edge. We call $N_r$ a {\it rooted partner} of $N_s$. By construction, $N_s$ has at most one pair of parallel edges.
This is precisely the case when $N_r$ has an underlying 3-cycle that contains the child of the root. Note that $N_s$ may have multiple rooted partners. A vertex $v$ in $N_s$ is called a {\it reticulation} if $N_s$ contains two edges, referred to as {\it reticulation edges}, that are directed into $v$. A {\it semi-directed level-$k$ network} is a semi-directed network $N_s$ such that a rooted partner of $N_s$ is a rooted binary level-$k$ network. Lastly,  an {\it unrooted binary phylogenetic $X$-tree} $T_u$ is a connected undirected acyclic graph whose leaf set is $X$ and whose inner vertices all have degree three. Note that an unrooted binary phylogenetic tree is a binary semi-directed network without reticulations. In the following, all rooted (resp. semi-directed) phylogenetic networks and rooted (resp. unrooted) phylogenetic trees are assumed to be binary. 

Now, let $T_u$ be an unrooted binary phylogenetic $X$-tree, and let $N_s$ be a semi-directed network on $X$. We say that $T_u$ is {\it displayed} by $N_s$ if there exists a subgraph $S$ of $N_s$ such that $S$ is a subdivision of $T_u$ (omitting the directions of the reticulation edges) and $S$ contains, for each reticulation $v$ in $N_s$, at most one reticulation edge incident with $v$. Similar to rooted phylogenetic networks, if $T_u$ is displayed by $N_s$, we call $S$ an {\it embedding} of $T_u$ in $N_s$. Furthermore, we refer to the set of all unrooted phylogenetic $X$-trees that are displayed by $N_s$ as the {\it unrooted display set} of $N_s$ and denote it by $D(N_s)$. 
 
Let $A=(f_1,f_2,\ldots,f_n)$ be an alignment of characters on $X$, and let $N_s$ be a semi-directed network on $X$. We define the {\it softwired parsimony score} of $A$ on $N_s$ as
\begin{align}\label{eq:soft-one-semi-directed}
PS_{\sw}(A,N_s)=\sum_{i=1}^n\min_{T_u\in D(N_s)}PS(f_i,T_u),
\end{align}
where the parsimony score of an unrooted phylogenetic $X$-tree $T$ is defined as in the rooted case since the corresponding concepts of the changing number and (minimum) extensions naturally translate to undirected graphs. 

The next lemma shows how the set of unrooted phylogenetic trees that are displayed by a semi-directed network $N_s$ is related to the set of rooted phylogenetic trees that are displayed by a rooted partner of $N_s$.

\begin{lemma}\label{l:semi-directed-display-set}
Let $N_s$ be a semi-directed network on $X$, and let $N_r$ be a rooted partner of $N_s$. Then, the following two statements hold. 
\begin{enumerate}[label={\upshape(\roman*)}]
\item If $T_u$ is an unrooted phylogenetic $X$-tree that is displayed by $N_s$, then there exists a rooted phylogenetic $X$-tree $T \in D(N_r)$ such that $T_u$ can be obtained from $T$ by deleting the root and suppressing its child. 
\item If $T$ is a rooted phylogenetic $X$-tree that is displayed by $N_r$, then there exists an unrooted phylogenetic $X$-tree $T_u \in D(N_s)$ such that $T_u$ can be obtained from $T$ by deleting the root and suppressing its child. 
\end{enumerate}
\end{lemma}

\begin{proof}
Throughout the proof, we assume that $N_s$ does not contain a pair of parallel edges. Indeed, if $N_s$ contains such a pair of edges, $e_1$ and $e_2$ say, we can obtain a semi-directed network $N_s'$ on $X$ from $N_s$, by deleting $e_1$, suppressing the two resulting degree-two vertices, and omitting the direction of $e_2$. Clearly $D(N_s)=D(N_s')$. Now, let $v$ be the child of the root $\rho$ in $N_r$, and let $w$ and $w'$ be the children of $v$. If either $w$ or $w'$ is a reticulation, we assume without loss of generality that $w'$ is a reticulation. Observe that at most one of $w$ and $w'$ is a reticulation. Each edge in $N_r$ that is not incident with $v$ corresponds to exactly one edge in $N_s$, and each edge in $N_s$ except $e_\rho = \{w,w'\}$ (resp.\  $e_\rho = (w,w')$ if $w'$ is a reticulation in $N_r$) corresponds to exactly one edge in~$N_r$. In particular, each reticulation edge in $N_r$ corresponds to exactly one such edge in $N_s$ 
and vice versa. 
First, let $T_u$ be an unrooted phylogenetic $X$-tree that is displayed by $N_s$. By definition, there exists an embedding $S_u$ of $T_u$ in $N_s$. If $e_\rho$ is an edge in $S_u$, we obtain an embedding $S$ of a rooted phylogenetic $X$-tree $T$ in $N_r$ from $S_u$ as follows. We replace $e_\rho$ with the three directed edges $(\rho,v)$, $(v,w)$ and $(v, w')$ and each edge $e \neq e_\rho$ with its directed counterpart in $N_r$.
By definition, $T$ is displayed by $N_r$, that is $T \in D(N_r)$. Further, by construction, $T_u$ can be obtained from $T$ by deleting the root and suppressing its child. Now, let us consider the case that $e_\rho$ is not an edge in $S_u$. Let $S'$ be the connected acyclic subgraph of $N_r$ that is obtained from $S_u$ by replacing each edge in $S_u$ with its directed counterpart in $N_r$. 
We next show that there is a unique vertex $v_s$ in $S'$ with in-degree zero and out-degree two. Since $S'$ is acyclic, it follows that $S'$ contains a vertex with in-degree zero. Clearly, the out-degree of this vertex cannot be one or three and, thus, $v_s$ exists. 
Assume towards a contradiction that $v_s'\neq v_s$ is another vertex with this property. As $S'$ is connected and does not contain a vertex with in-degree two, one of $v_s'$ and $v_s$ is a descendant of the other in $S'$. But then one of these vertices has in-degree one, a contradiction.
It now follows that each vertex in $S'$ is a descendant of $v_s$. Let $\pi = (\rho = v_1,v_2,\ldots,v_t=v_s)$ be a directed path in $N_r$. Such a path $\pi$ exists as there is a directed path from the root to any vertex in $N_r$. Since each vertex in $S'$ is a descendant of $v_s$, it follows that $v_i$ with $i < t$ is not in $S'$. Then, we obtain an embedding $S$ of a rooted phylogenetic $X$-tree $T$ in $N_r$ from $S'$ by adding the $t-1$ directed edges $(v_i, v_{i+1})$, $1 \leq i < t$, to $S'$. Again, we have $T \in D(N_r)$ and, by construction, $T_u$ can be obtained from $T$ by deleting the root and suppressing its child. Hence, (i) holds.

Second, let $T$ be a rooted phylogenetic $X$-tree that is displayed by $N_r$. Let $S$ be an embedding of $T$ in $N_r$, and let $\pi = (\rho=v_1,v_2,\ldots, v_t)$ be the root path of $S$. Then, we obtain an embedding  $S_u$ of an unrooted phylogenetic $X$-tree $T_u$ in $N_s$ from $S$ by deleting each vertex that lies on $\pi$ except for $v_t$, turning directed edges into undirected ones and, if $t=2$, suppressing $v_t$. If $t=2$, then $\pi$ consists of the single edge $(\rho, v)$, and $S$ contains $(v,w)$ and $(v,w')$. Hence, $S_u$ contains the edge $\{w,w'\}$ as a result of suppressing $v$. Otherwise, if $t\geq 3$, then $S$ contains $(\rho,v)$ and exactly one of $(v,w)$ and $(v,w')$, and $v_t$ is a vertex in $N_s$. Furthermore, if $t \geq 4$, each edge $(v_i, v_{i+1})$ with $3 \leq i < t$ that is traversed by $\pi$ corresponds to a unique edge consisting of the same vertices in $N_s$. It now follows that $S_u$ is indeed an embedding of $T_u$ in $N_s$, that is, $T_u \in D(N_s)$. By construction, $T_u$ can be obtained from $T$ by deleting the root and suppressing its child. Hence, (ii) holds as well.  
\end{proof}

We next obtain the following corollary from Lemma~\ref{l:semi-directed-display-set}.

\begin{corollary}\label{c:pars-semi-directed}
Let $N_s$ be a semi-directed network on $X$, let $N_r$ be a rooted partner of $N_s$, and let $A$ be an alignment of characters on $X$. Then $PS_{\sw}(A,N_s) = PS_{\sw}(A,N_r)$. 
\end{corollary}
\begin{proof}
The corollary follows from Lemma~\ref{l:semi-directed-display-set} and the fact that, if $T$ is a rooted phylogenetic $X$-tree and $T_u$ is an unrooted phylogenetic $X$-tree such that $T_u$ can be obtained from $T$ by deleting the root and suppressing its child, then $PS(A,T)=PS(A,T_u)$.
\end{proof}

We are now in a position to state the main result of the section. 

\begin{theorem}\label{t:main-semi-directed}
Let $N_s$ be a semi-directed network on $X$, and let $T_u$ be an unrooted phylogenetic $X$-tree that is displayed by $N_s$. Furthermore, let $A$ be an alignment of characters on $X$. Then 
$$PS(A,T_u)\leq (k+1)\cdot PS_{\sw}(A,N_s),$$ 
where $k$ is the level of $N_s$. 
\end{theorem}

\begin{proof}
Let $N_r$ be a rooted partner of $N_s$, and let $T \in D(N_r)$ be a rooted phylogenetic $X$-tree such that deleting the root in $T$ and suppressing its child yields $T_u$. The rooted phylogenetic $X$-tree $T$ exists by Lemma~\ref{l:semi-directed-display-set} (i). Since an unrooted phylogenetic tree is a semi-directed network without reticulations, we have (i) $PS(A, T_u) = PS(A, T)$ by Corollary~\ref{c:pars-semi-directed}. Furthermore, by Theorem~\ref{t:main}, we have (ii) $PS(A, T) \leq (k+1)\cdot PS_{\sw}(A,N_r)$. Finally, by Corollary~\ref{c:pars-semi-directed}, we have (iii) $PS_{\sw}(A,N_r) = PS_{\sw}(A,N_s)$. Combining (i)--(iii) yields $PS(A, T_u) \leq (k+1)\cdot PS_{\sw}(A,N_s)$.
\end{proof}

We now shift our focus from semi-directed networks to \emph{unrooted} phylogenetic networks.
An {\it unrooted binary phylogenetic network} $U$ on $X$ is a connected undirected graph without any loops or edges in parallel, whose leaf set is $X$, and whose inner vertices all have degree three. As before,  we omit the term binary in the following as all unrooted phylogenetic networks considered in this section are binary. 

Let $U$ be an unrooted phylogenetic network on $X$. We say that an unrooted phylogenetic $X$-tree $T_u$ is {\it displayed} by $U$ if there exists a subgraph of $U$ that is a subdivision of $T_u$. Furthermore, we refer to the set of all unrooted phylogenetic $X$-trees that are displayed by $U$ as the {\it display set} of $U$ and denote it by $D(U)$. We call $U$  an {\it unrooted level-$k$ network} if at most $k$ edges have to be deleted in each biconnected component of $U$ such that the resulting graph is acyclic. Lastly, if $U$ can be obtained from a rooted phylogenetic network $N$ by deleting its root,  suppressing the child of the root, and omitting all edge directions, we say that $N$ is an {\it orientation} of $U$. 

Let $A=(f_1,f_2,\ldots,f_n)$ be an alignment of characters on $X$, and let $U$ be an unrooted phylogenetic network on $X$. We define the {\it softwired parsimony score} of $A$ on $U$ as
\begin{align}\label{eq:soft-unrooted}
PS_{\sw}(A,U)=\sum_{i=1}^n\min_{T\in D(U)}PS(f_i,T).
\end{align}

Next, we present an unrooted level-1 network $U$ and a binary character $f$ such that $PS_{\sw}(f, U) \neq PS_{\sw}(f, N)$ for an orientation $N$ of $U$, that is, we show that Corollary~\ref{c:pars-semi-directed} does not translate  from semi-directed networks to unrooted phylogenetic networks.  Moreover, we give two different orientations $N$ and $N'$ of $U$ with $PS_{\sw}(f, N) \neq PS_{\sw}(f, N')$. 

\begin{figure}[t]
\centering
\includegraphics[width=.8\textwidth]{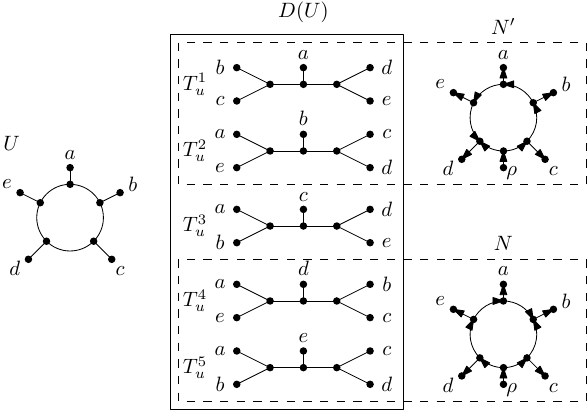}
\caption{
An unrooted level-1 network $U$ (left) with its display set $D(U)$ (middle), and two orientations $N$ and $N'$ of $U$ (right) with their respective display sets indicated by a dashed rectangle. More precisely, $D(N)$ and $D(N')$ can be obtained by orienting the enclosed unrooted phylogenetic $X$-trees.}
\label{fig:displaysets-orientations}
\end{figure}
To this end, let $U$ be the unrooted level-1 network on $X = \{a,b,c,d,e\}$, and let $N$ and $N'$ be the two orientations of $U$ as shown in  Figure~\ref{fig:displaysets-orientations}. The display set $D(U)$ of size five is shown in the middle of the same figure. By deleting the root and suppressing its child in each element of $D(N)$ (resp. $D(N')$), we obtain a subset of $D(U)$ as indicated by the two dashed rectangles that each enclose $N$ (resp. $N'$) and two elements of $D(U)$.  
Now for the single binary character $f\colon X \rightarrow \{0,1\}$ with 
\[
f(a) = f(b) = f(c) = 0 \text{ and } f(d) = f(e) = 1, 
\]
we have $PS_{\sw}(f, U) = 1$, $PS_{\sw}(f,N) = 2$ and $PS_{\sw}(f,N') = 1$. Since the softwired parsimony score of an unrooted phylogenetic network $U$ is not necessarily the same as the score of an orientation of $U$,  we cannot represent an unrooted phylogenetic network by an arbitrary orientation in the way we used a rooted partner of a semi-directed network to obtain Theorem~\ref{t:main-semi-directed}. While our example shows that using the same approach as in the semi-directed setting is not viable, it does not exclude the existence of similar bounds for unrooted phylogenetic networks or classes thereof (such as unrooted level-1 networks, for example).

\section{Concluding remarks} \label{sec:conclusion}
In this paper, we have obtained a bound on the softwired parsimony score of a gap-free alignment of multistate characters on rooted as well as semi-directed phylogenetic level-$k$ networks. To be precise, we have shown that the maximum difference between the softwired parsimony score of a phylogenetic network $N$ and the parsimony score of any tree displayed by $N$ is bounded by $k+1$ times the parsimony score of $N$. Unfortunately, our approximation result as stated in Theorem~\ref{t:main} cannot be generalised to alignments with gaps since it was already shown in~\cite[Corollary 2]{kelk2019finding} that computing the softwired parsimony score of a level-$1$ network for an alignment of binary characters that additionally allows gaps is APX-hard.

Extending the notion of softwired parsimony to semi-directed networks and exploiting a connection between the display sets of semi-directed networks and their rooted partners, we have shown that an analogous bound holds for semi-directed networks. For unrooted networks, on the other hand, the approach via rooted partners (more formally, via orientations) does not seem to be viable. Nevertheless, it would be an interesting question for future research to investigate if an analogous or similar bound for the softwired parsimony score can be obtained in some other way for unrooted phylogenetic networks.

Another interesting direction for future research would be to analyse whether our results also apply in the case of {\it non-binary} phylogenetic networks, i.e., phylogenetic networks that may have vertices of degree strictly greater than three. While there exists a polynomial time algorithm to compute the parsimony score of a given non-binary phylogenetic tree with character states assigned to its leaves, namely the Fitch-Hartigan algorithm~\cite{fitch1971toward,hartigan1973algorithm}, other concepts that our results rely on, such as the rSPR distance and its relation to the switching distance, have been studied less for non-binary phylogenetic trees (see \cite[Section~2]{fischer2016maximum} for a related discussion). \\

\noindent {\bf Acknowledgements.} The first and second authors thank the New Zealand Marsden Fund for financial support. All authors thank Steven Kelk for helpful discussions.

\bibliographystyle{abbrvnat}
\bibliography{References.bib}

\begin{thebibliography}{25}
\providecommand{\natexlab}[1]{#1}
\providecommand{\url}[1]{\texttt{#1}}
\expandafter\ifx\csname urlstyle\endcsname\relax
  \providecommand{\doi}[1]{doi: #1}\else
  \providecommand{\doi}{doi: \begingroup \urlstyle{rm}\Url}\fi

\bibitem[Allman et~al.(2019)Allman, Ba{\~n}os, and Rhodes]{allman2019nanuq}
E.~S. Allman, H.~Ba{\~n}os, and J.~A. Rhodes.
\newblock {NANUQ}: {A} method for inferring species networks from gene trees
  under the coalescent model.
\newblock \emph{Algorithms for Molecular Biology}, 14:\penalty0 1--25, 2019.

\bibitem[Cardona et~al.(2009)Cardona, Rossello, and
  Valiente]{cardona2009comparison}
G.~Cardona, F.~Rossello, and G.~Valiente.
\newblock Comparison of tree-child phylogenetic networks.
\newblock \emph{IEEE/ACM Transactions on Computational Biology and
  Bioinformatics}, 6\penalty0 (4):\penalty0 552--569, 2009.

\bibitem[D{\"o}cker et~al.(2024)D{\"o}cker, Linz, and
  Semple]{docker2024hypercubes}
J.~D{\"o}cker, S.~Linz, and C.~Semple.
\newblock Hypercubes and {H}amilton cycles of display sets of rooted
  phylogenetic networks.
\newblock \emph{Advances in Applied Mathematics}, 152:\penalty0 102595, 2024.

\bibitem[Felsenstein(2004)]{felsenstein2004inferring}
J.~Felsenstein.
\newblock \emph{Inferring Phylogenies}.
\newblock Sinauer Associates, US, 2004.

\bibitem[Fischer and Kelk(2016)]{fischer2016maximum}
M.~Fischer and S.~Kelk.
\newblock On the maximum parsimony distance between phylogenetic trees.
\newblock \emph{Annals of Combinatorics}, 20:\penalty0 87--113, 2016.

\bibitem[Fischer et~al.(2015)Fischer, van Iersel, Kelk, and
  Scornavacca]{fischer2015computing}
M.~Fischer, L.~van Iersel, S.~Kelk, and C.~Scornavacca.
\newblock On computing the maximum parsimony score of a phylogenetic network.
\newblock \emph{SIAM Journal on Discrete Mathematics}, 29\penalty0
  (1):\penalty0 559--585, 2015.

\bibitem[Fitch(1971)]{fitch1971toward}
W.~M. Fitch.
\newblock Toward defining the course of evolution: minimum change for a
  specific tree topology.
\newblock \emph{Systematic Biology}, 20\penalty0 (4):\penalty0 406--416, 1971.

\bibitem[Frohn and Kelk(2023)]{frohn23}
M.~Frohn and S.~Kelk.
\newblock A $2$-approximation algorithm for the softwired parsimony problem on
  binary, tree-child phylogenetic networks.
\newblock \emph{submitted}, 2023.

\bibitem[Gross and Long(2018)]{gross2018distinguishing}
E.~Gross and C.~Long.
\newblock Distinguishing phylogenetic networks.
\newblock \emph{SIAM Journal on Applied Algebra and Geometry}, 2\penalty0
  (1):\penalty0 72--93, 2018.

\bibitem[Hartigan(1973)]{hartigan1973algorithm}
J.~A. Hartigan.
\newblock Minimum mutation fits to a given tree.
\newblock \emph{Biometrics}, 29\penalty0 (1):\penalty0 53, 1973.

\bibitem[Hollering and Sullivant(2021)]{hollering2021identifiability}
B.~Hollering and S.~Sullivant.
\newblock Identifiability in phylogenetics using algebraic matroids.
\newblock \emph{Journal of Symbolic Computation}, 104:\penalty0 142--158, 2021.

\bibitem[Huber et~al.(2016)Huber, Moulton, Steel, and Wu]{huber2016weakly}
K.~T. Huber, V.~Moulton, M.~Steel, and T.~Wu.
\newblock Folding and unfolding phylogenetic trees and networks.
\newblock \emph{Journal of Mathematical Biology}, 73\penalty0 (6–7):\penalty0
  1761--1780, 2016.

\bibitem[Kannan and Wheeler(2012)]{kannan2012maximum}
L.~Kannan and W.~C. Wheeler.
\newblock Maximum parsimony on phylogenetic networks.
\newblock \emph{Algorithms for Molecular Biology}, 7:\penalty0 1--10, 2012.

\bibitem[Kelk and Fischer(2017)]{kelk2017complexity}
S.~Kelk and M.~Fischer.
\newblock On the complexity of computing \mbox{MP} distance between binary
  phylogenetic trees.
\newblock \emph{Annals of Combinatorics}, 21:\penalty0 573--604, 2017.

\bibitem[Kelk et~al.(2019)Kelk, Pardi, Scornavacca, and van
  Iersel]{kelk2019finding}
S.~Kelk, F.~Pardi, C.~Scornavacca, and L.~van Iersel.
\newblock Finding a most parsimonious or likely tree in a network with respect
  to an alignment.
\newblock \emph{Journal of Mathematical Biology}, 78:\penalty0 527--547, 2019.

\bibitem[Nakhleh et~al.(2005)Nakhleh, Jin, Zhao, and
  Mellor-Crummey]{nakhleh2005reconstructing}
L.~Nakhleh, G.~Jin, F.~Zhao, and J.~Mellor-Crummey.
\newblock Reconstructing phylogenetic networks using maximum parsimony.
\newblock In \emph{2005 IEEE Computational Systems Bioinformatics Conference
  (CSB'05)}, pages 93--102. IEEE, 2005.

\bibitem[Sansom et~al.(2018)Sansom, Choate, Keating, and
  Randle]{sansom2018parsimony}
R.~S. Sansom, P.~G. Choate, J.~N. Keating, and E.~Randle.
\newblock Parsimony, not {B}ayesian analysis, recovers more stratigraphically
  congruent phylogenetic trees.
\newblock \emph{Biology Letters}, 14\penalty0 (6):\penalty0 20180263, 2018.

\bibitem[Schrago et~al.(2018)Schrago, Aguiar, and
  Mello]{schrago2018comparative}
C.~G. Schrago, B.~O. Aguiar, and B.~Mello.
\newblock Comparative evaluation of maximum parsimony and {B}ayesian
  phylogenetic reconstruction using empirical morphological data.
\newblock \emph{Journal of Evolutionary Biology}, 31\penalty0 (10):\penalty0
  1477–1484, 2018.

\bibitem[Semple(2007)]{Semple:2007ug}
C.~Semple.
\newblock {Hybridization networks}.
\newblock In O.~Gascuel and M.~Steel, editors, \emph{Reconstructing Evolution:
  New Mathematical and Computational Advances}, pages 277--314. Oxford
  University Press, UK, 2007.

\bibitem[Smith(2019)]{smith2019bayesian}
M.~R. Smith.
\newblock Bayesian and parsimony approaches reconstruct informative trees from
  simulated morphological datasets.
\newblock \emph{Biology Letters}, 15\penalty0 (2):\penalty0 20180632, 2019.

\bibitem[Sol{\'\i}s-Lemus and An{\'e}(2016)]{solis2016inferring}
C.~Sol{\'\i}s-Lemus and C.~An{\'e}.
\newblock Inferring phylogenetic networks with maximum pseudolikelihood under
  incomplete lineage sorting.
\newblock \emph{PLoS Genetics}, 12\penalty0 (3):\penalty0 e1005896, 2016.

\bibitem[Sol{\'\i}s-Lemus et~al.(2017)Sol{\'\i}s-Lemus, Bastide, and
  An{\'e}]{solis2017phylonetworks}
C.~Sol{\'\i}s-Lemus, P.~Bastide, and C.~An{\'e}.
\newblock \mbox{PhyloNetworks:} a package for phylogenetic networks.
\newblock \emph{Molecular Biology and Evolution}, 34\penalty0 (12):\penalty0
  3292--3298, 2017.

\bibitem[Stamatakis et~al.(2005)Stamatakis, Ludwig, and
  Meier]{stamatakis2005raxml}
A.~Stamatakis, T.~Ludwig, and H.~Meier.
\newblock \mbox{RAxML-III:} {A} fast program for maximum likelihood-based
  inference of large phylogenetic trees.
\newblock \emph{Bioinformatics}, 21\penalty0 (4):\penalty0 456--463, 2005.

\bibitem[van Iersel et~al.(2018)van Iersel, Jones, and
  Scornavacca]{van2018improved}
L.~van Iersel, M.~Jones, and C.~Scornavacca.
\newblock Improved maximum parsimony models for phylogenetic networks.
\newblock \emph{Systematic Biology}, 67\penalty0 (3):\penalty0 518--542, 2018.

\bibitem[Zhang et~al.(2020)Zhang, Huelsenbeck, and
  Ronquist]{zhang2020parsimonyguided}
C.~Zhang, J.~P. Huelsenbeck, and F.~Ronquist.
\newblock Using parsimony-guided tree proposals to accelerate convergence in
  {B}ayesian phylogenetic inference.
\newblock \emph{Systematic Biology}, 69\penalty0 (5):\penalty0 1016–1032,
  2020.

\end{thebibliography}

\end{document}